\documentclass[10pt,journal,twoside,twocolumn]{IEEEtran}

\IEEEoverridecommandlockouts
\usepackage{xspace,exscale}
\usepackage{fancybox}
\usepackage{graphicx}
\usepackage{color}
\usepackage{amsfonts}
\usepackage{amsthm}
\usepackage{amssymb}
\usepackage{url}

\usepackage[noadjust]{cite}

\usepackage{amsmath}
	\makeatletter
	\let\over=\@@over \let\overwithdelims=\@@overwithdelims
	\let\atop=\@@atop \let\atopwithdelims=\@@atopwithdelims
  	\let\above=\@@above \let\abovewithdelims=\@@abovewithdelims
  	\makeatother
\usepackage{fixltx2e}

\usepackage{rotating}

\usepackage{ifpdf}

\usepackage{subfigure}
\usepackage{psfrag}

\newcommand{\matt}{\ensuremath{\mathcal{T}}}
\newcommand{\matc}{\ensuremath{\mathcal{C}}}

\newcommand{\FF}{\ensuremath{\mathbb{F}}}

\ifx\eqref\undefined
	\newcommand{\eqref}[1]{~(\ref{#1})}
\fi
\ifx\mod\undefined
	\def\mod{\mathop{\rm mod}}
\fi

\def\exp{\mathop{\rm exp}}

\def\EE{\mathbb{E}\,}

\def\PP{\mathbb{P}}

\def\eqdef{\stackrel{\triangle}{=}}

\def\rad{\mathop{\rm rad}}
\def\arad{\mathop{\overline{\rm rad}}}

\def\unifto{\mathop{{\mskip 3mu plus 2mu minus 1mu%
	\setbox0=\hbox{$\mathchar"3221$}%
	\raise.6ex\copy0\kern-\wd0%
	\lower0.5ex\hbox{$\mathchar"3221$}}\mskip 3mu plus 2mu minus 1mu}}

\ifx\lesssim\undefined
\def\simleq{{{\mskip 3mu plus 2mu minus 1mu%
	\setbox0=\hbox{$\mathchar"013C$}%
	\raise.2ex\copy0\kern-\wd0%
	\lower0.9ex\hbox{$\mathchar"0218$}}\mskip 3mu plus 2mu minus 1mu}}
\else
\def\simleq{\lesssim}
\fi

\ifx\gtrsim\undefined
\def\simgeq{{{\mskip 3mu plus 2mu minus 1mu%
	\setbox0=\hbox{$\mathchar"013E$}%
	\raise.2ex\copy0\kern-\wd0%
	\lower0.9ex\hbox{$\mathchar"0218$}}\mskip 3mu plus 2mu minus 1mu}}
\else
\def\simgeq{\gtrsim}
\fi

\newif\ifmapx
{\catcode`/=0 \catcode`\\=12/gdef/mkillslash\#1{#1}}
\edef\jobnametmp{\expandafter\string\csname listconv_apx\endcsname}
\edef\jobnameapx{\expandafter\mkillslash\jobnametmp}
\edef\jobnameexpand{\jobname}
\ifx\jobnameexpand\jobnameapx
\mapxtrue
\else
\mapxfalse
\fi

\long\def\apxonly#1{\ifmapx{\color{blue}#1}\fi}

\newtheorem{theorem}{Theorem}
\newtheorem{lemma}[theorem]{Lemma}
\newtheorem{corollary}[theorem]{Corollary}

\newtheorem{remark}{Remark}

\begin{document}

\title{Upper bound on list-decoding radius of binary codes}

\author{Yury Polyanskiy
\thanks{YP is with the Department of Electrical Engineering 
and Computer Science, MIT, Cambridge, MA 02139 USA.
	\mbox{e-mail:~{\ttfamily yp@mit.edu}.}}%
\thanks{
The research was supported by the NSF grant CCF-13-18620 and NSF Center for Science of Information (CSoI) 
under grant agreement CCF-09-39370. This work was presented at 2015 IEEE International Symposium on Information Theory
(ISIT), Hong Kong, CN, Jun 2015.}}
%

\maketitle

\begin{abstract} Consider the problem of packing Hamming balls of a given relative radius subject to the
constraint that they cover any point of the ambient Hamming space with multiplicity at most $L$. For odd $L\ge 3$ 
an asymptotic upper bound on the rate of any such packing is proven. The resulting bound improves the best known bound
(due to Blinovsky'1986) for rates below a certain threshold. The method is a superposition of the linear-programming
idea of Ashikhmin, Barg and Litsyn (that was used previously to improve the estimates of Blinovsky for $L=2$) and a
Ramsey-theoretic technique of Blinovsky. As an application it is shown that for all odd $L$ the slope of the rate-radius
tradeoff is zero at zero rate.
\end{abstract}


%
\IEEEpeerreviewmaketitle


\section{Main result and discussion}

One of the most well-studied problems in information theory asks to find the maximal rate at which codewords can be
packed in binary space with a given minimum distance between codewords. Operationally, this (still unknown) rate gives
the capacity of the binary input-output channel subject to adversarial noise of a given level. A natural generalization
was considered by Elias and Wozencraft~\cite{PE57,JW58}, who allowed the decoder to output a list of size $L$. In
this paper we provide improved upper bounds on the latter question.

Our interest in bounding the asymptotic tradeoff for the list-decoding problem is motivated by our study of fundamental
limits of joint source-channel communication~\cite{KMP12isit}. Namely, in~\cite[Theorem 6]{YP15-cjscc_isit} we proposed
an extension of the previous result in \cite[Theorem 7]{KMP12isit} that required bounding rate for the
list-decoding problem.

We proceed to formal definitions and brief overview of known results. 
For a binary code $\matc \subset \FF_2^n$ we define its list-size $L$ decoding radius as
$$ \tau_L(\matc) \eqdef {1\over n} \max\{r: \forall x\in\FF_2^n \,\, |\matc \cap \{x+B_r^n\}| \le L \}\,, $$
where Hamming ball $B_r^n$ and Hamming sphere $S_r^n$ are defined as
\begin{align} 
B_r^n &\eqdef \{x\in \FF_2^n: |x| \le r\}\,, \\
S_r^n &\eqdef \{x\in \FF_2^n: |x| = r\} 
\end{align}
with $|x| = |\{i: x_i = 1\}|$ denoting the Hamming weight of $x$. Alternatively, we may define $\tau_L$ as
follows:\footnote{${\matc \choose j}$ denotes the set of all subsets of $\matc$ of size $j$.}
$$ \tau_L(\matc) = {1\over n} \left( \min\left\{ \rad(S): S\in {\matc \choose L+1}\right\} - 1
\right)\,,$$
where $\rad(S)$ denotes radius of the smallest ball containing $S$ (known as Chebyshev radius):
$$ \rad(S) \eqdef \min_{y\in\FF_2^n} \max_{x\in S} |y-x|\,. $$
	
The asymptotic tradeoff between rate and list-decoding radius $\tau_L$ is defined as usual:
\begin{align} \tau_L^*(R) &\eqdef \limsup_{n\to\infty} \max_{\matc: |\matc| \ge 2^{nR}} \tau_L(\matc) \\
   R_L^*(\tau) &\eqdef \limsup_{n\to\infty} \max_{\matc: \tau_L(\matc) \ge \tau} {1\over n} \log |\matc| 
\end{align}

The best known upper (converse) bounds on this tradeoff are as follows:
\begin{itemize}
\item List size $L=1$: The best bound to date was found by McEliece, Rodemich, Rumsey and
	Welch~\cite{MRRW77}:
	\begin{align} R_1^*(\tau) &\le R_{LP2}(2\tau)\,,\\
		R_{LP2}(\delta) &\eqdef \min \log 2 -h(\alpha) + h(\beta)\,, 
	\end{align}		
	where $h(x)=-x \log x - (1-x) \log(1-x)$ and minimum is taken over all $0\le\beta\le\alpha\le1/2$ satisfying
		$$ 2{\alpha (1-\alpha) - \beta(1-\beta)\over 1+2\sqrt{\beta (1-\beta)}} \le \delta $$
	For rates $R<0.305$ this bound coincides with the simpler bound:
	\begin{align} \tau_1^*(R) &\le {1\over 2} \delta_{LP1}(R)\,,\\
	 \delta_{LP1}(R) &\eqdef {1\over 2} - \sqrt{\beta(1-\beta)}\,, \quad R=\log 2-h(\beta)\,,  
	\end{align}
	where $\beta\in[0,{1\over2}]$.
\item List size $L=2$: The bound found by Ashikhmin, Barg and Litsyn~\cite{ashikhmin2000new} is given
	as%
	\footnote{This result follows from \apxonly{(fixing typos and)} optimizing~\cite[Theorem 4]{ashikhmin2000new}. It is
	slightly stronger than what is given in~\cite[Corollary 5]{ashikhmin2000new}.\apxonly{The improvement is due
	to the fact that taking $\delta = \delta_{LP2}(R)$ is not always the best for Theorem 4. Taking
	$\delta>\delta_{LP2}$ yields about $10^{-3}$ improvement in rate in the range $\tau>\tau_0$.}}%
	$$ R_2^*(\tau) \le \log 2-h(2\tau)+R_{up}(2\tau, 2\tau)\,,$$
	where $R_{up}(\delta, \alpha)$ is the best known upper bound on rate of codes with minimal distance $\delta n$
	constrained to live on Hamming spheres $S_{\alpha n}^n$. The expression for $R_{up}(\delta, \alpha)$ can be
	obtained by using the linear programming bound from~\cite{MRRW77} and applying Levenshtein's monotonicity,
	cf.~\cite[Lemma 4.2(6)]{AS01}. The resulting expression is
	\begin{equation} R_2^*(\tau) \le \begin{cases} R_{LP2}(2\tau)\,, &\tau \le \tau_0\\
	\log 2-h(2\tau) + h(u(\tau)), & \tau > \tau_0\,,\end{cases} 
	\label{eq:abl2}
	\end{equation}	
	where $\tau_0 \approx 0.1093$ and 
	$$ u(\tau)={1\over2} - \sqrt{{1\over 4} - (\sqrt{\tau-3\tau^2}-\tau)^2}$$
	(cf.~\cite[(9)]{AS01}).
\item For list sizes $L\ge 3$: The original bound of Blinovsky~\cite{blinovsky1986bounds} appears to be  the best 
(before this work):
\begin{equation}\label{eq:blin}
		\tau_L^*(R) \le \sum_{i=1}^{\lceil L/2 \rceil} {{2i-2\choose i-1}\over i} (\lambda(1-\lambda))^i\,, \qquad
	R=1-h(\lambda)\,,  
\end{equation}
	where $\lambda \in [0,{1\over2}]$.
	Note that~\cite{blinovsky1986bounds} also gives a non-constructive lower bound on $\tau_L^*(R)$.
Results on list-decoding over non-binary alphabets are also known, see~\cite{blinovsky2005code,guruswami2005lower}.
\end{itemize}

In this paper we improve the bound of Blinovsky for lists of odd size and rates below a certain threshold. To that
end we will mix the ideas of Ashikhmin, Barg and Litsyn (namely, extraction of a large spectrum component from the code) and
those of Blinovsky (namely, a Ramsey-theoretic reduction to study of symmetric subcodes).

To present our main result, we need to define exponent of Krawtchouk polynomial $K_{\beta n}(\xi n)=\exp\{n
E_\beta(\xi) + o(n)\}$. For $\xi\in[0,{1\over2}-\sqrt{\beta(1-\beta)}]$ the value of $E_\beta(\xi)$ was found
in~\cite{KL95}. Here we give it in the following parametric form, cf.~\cite{IS98} or~\cite[Lemma 4]{YP13}:
\begin{align}\label{eq:ed_2}
	E_\beta(\xi) &= \xi \log(1-\omega) + (1-\xi) \log(1+\omega) - \beta \log \omega \\
	 \xi &= {1\over2}( 1-(1-\beta)\omega - \beta \omega^{-1})\,,
\end{align}
where 
$$ \omega \in \left[{\beta\over 1-\beta}, \sqrt{\beta\over 1-\beta}\right]\,. $$

Our main result is the following:
\begin{theorem}\label{th:main}
	Fix list size $L\ge2$, rate $R$ and an arbitrary $\beta \in [0, 1/2]$ with $h(\beta)\le R$. Then any sequence of codes $\matc_n\subset \{0,1\}^n$ of rate $R$ satisfies
\begin{multline}\label{eq:thmain}
			\limsup_{n\to\infty} \tau_{L}(\matc_n) \le \\\max_{j,\xi_0} \xi_0 g_{j}\left(1-{\xi_1\over 2\xi_0}\right) 
			+ (1-\xi_0) g_j\left(\xi_1\over 2(1-\xi_0)\right)\,, 
\end{multline}			
where maximization is over $\xi_0$ satisfying
\begin{equation}\label{eq:xi0ineq}
		0 \le \xi_0 \le {1\over2} - \sqrt{\beta(1-\beta)} 
\end{equation}
and $j$ ranging over $\{0, 1, 3, \ldots, 2k+1, \ldots, L\}$ if $L$ is odd and over $\{0,
2, \ldots,2k, \ldots L\}$ if $L$ is even. Quantity $\xi_1 = \xi_1(\xi_0, \delta, R)$ is a unique solution of
\begin{multline}\label{eq:xidef}
		R + h(\beta) - 2E_\beta(\xi_0) = \\
			h(\xi_0) - \xi_0 h\left(\xi_1\over 2\xi_0\right) - (1-\xi_0) h\left(\xi_1\over 2(1-\xi_0)\right)\,,
\end{multline}
	on the interval $[0, 2\xi_0 (1-\xi_0)]$ and functions $g_j(\nu)$ are defined as
	\begin{equation}\label{eq:gdef}
			g_j(\nu) \eqdef {1\over L+j} \left(L\nu - \EE[|2W-L-j|^+]\right)\,, W \sim \mathrm{Bino}(L, \nu) 
\end{equation}	
\end{theorem}

As usual with bounds of this type, cf.~\cite{BM05}, it appears that taking $h(\beta)=R$ can be done without loss. Under
such choice, our bound outperforms Blinovsky's for all odd $L$ and all rates small enough (see Corollary~\ref{th:slope} below). The bound for $L=3$
is compared in Fig.~\ref{fig:comp} with the result of Blinovsky numerically. For larger odd $L$
the comparison is similar, but the range of rates where our bound outperforms Blinovsky's becomes smaller, see
Table~\ref{tab:comp}.

Evaluation of Theorem~\ref{th:main} is computationally possible, but is somewhat tedious. Fortunately, for small $L$ the
maximum over $\xi_0$ and $j$ is attained at $\xi_0 = {1\over2}-\sqrt{\beta(1-\beta)}$ and $j=1$. We rigorously prove
this for $L=3$:%
\footnote{Notice that proofs of each of
the two Corollaries below contain different relaxations of the bound~\eqref{eq:thmain}, e.g.~\eqref{eq:slo1}, which are
easier to evaluate. Notice also that in Table~\ref{tab:comp} for the last two entries ($L=9,11$) at the high endpoint of
rate the maximum over $\xi_0$ is attained \textit{not} at ${1\over2}-\sqrt{\beta(1-\beta)}$.}

\begin{corollary}\label{th:l3}
	For list-size $L=3$ we have
	\begin{equation} 
	\tau_L^*(R) \le {3\over 4} \delta - {1\over 16}\left({(2\delta - \xi_1)^3\over \delta^2} + {\xi_1^3\over (1-\delta)^2}
	\right)\,,\label{eq:bd3b}
	\end{equation}	
	where $\delta\in(0,1/2]$ and $\xi_1 \in [0,2\delta(1-\delta)]$ are functions of $R$ determined from
	\begin{align} R &= h\left({1\over2} - \sqrt{\delta (1-\delta)}\right)\label{eq:bd3a}\,,\\
		R &= \log 2 - \delta h\left(\xi_1\over 2\delta\right) - (1-\delta) h\left(\xi_1\over
		2(1-\delta)\right) \label{eq:bd3c}
	\end{align}	
\end{corollary}

\begin{figure}
	\centering
	\includegraphics[width=.4\textwidth]{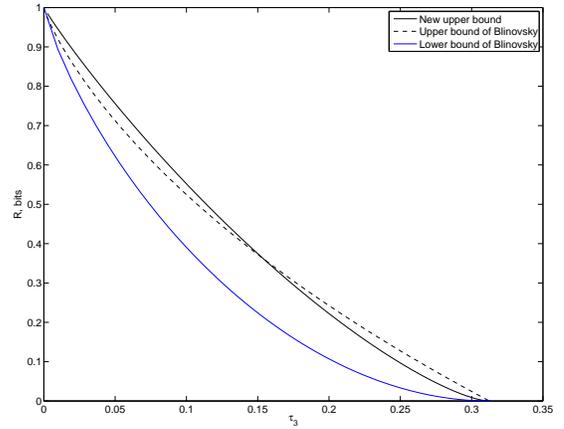}
	
	\caption{Comparison of bounds on $R^*_L(\tau)$ for list size $L=3$}\label{fig:comp}
\end{figure}

\begin{table}
	\caption{Rates for which new bound$^*$ improves state of the art}\label{tab:comp}
	\centering
	\begin{tabular}{l|l}
		List size $L$ & Range of rates\\
		\hline
		$L=3$ & $0  < R \le 0.361$\\
		$L=5$ & $0 < R \le 0.248$\\
		$L=7$ & $0 < R \le 0.184$\\
		$L=9$ & $0 < R \le 0.136$\\
		$L=11$ & $0 < R \le 0.100$
	\end{tabular}\\[5pt]
	$^*$ This is computation of~\eqref{eq:thmain} with $h(\beta)=R$.
	\apxonly{Note: Weaker bound~\eqref{eq:slo1} has much worse rate-limits: $0.115 (L=3), 0.042 (L=5), 0.020
	(L=7)$.}
\end{table}

Another interesting implication of Theorem~\ref{th:main} is that it allows us to settle the question of slope of the
curve $R_L^*(\tau)$ at zero rate. Notice that Blinovsky's converse bound~\eqref{eq:blin} has a negative slope, while his
achievability bound has a zero slope. Our bound always has a zero slope for odd $L$ (but not for even $L$, see
Remark~\ref{rm:evenL} in Section~\ref{sec:slope}):

\begin{corollary}\label{th:slope} Fix arbitrary odd $L\ge 3$. There exists $R_0 = R_0(L) > 0$ such that for all rates $R < R_0$ we have
	\begin{equation}\label{eq:sl0}
			\tau^*_L(R) \le g_1(\delta_{LP1}(R))\,,
	\end{equation}		
	where $g_1(\cdot)$ is a degree-$L$ polynomial defined in~\eqref{eq:gdef}.
		In particular, 
\begin{equation}\label{eq:zeroslope}
	\left.{d\over d\tau}\right|_{\tau=\tau^*_L(0)} R_L^*(\tau) = 0\,,
\end{equation}	
	where the zero-rate radius is $\tau^*_L(0) = {1\over 2} - 2^{-L-1}{L\choose {L-1\over2}}$.
\end{corollary}

Before closing our discussion we make some additional remarks:
\begin{enumerate}
	\item The bound in Theorem~\ref{th:main} can be slightly improved by replacing $\delta_{LP1}(R)$, that appears
	in the right-hand side of~\eqref{eq:xi0ineq}, with a 
	better bound, a so-called second linear-programming bound $\delta_{LP2}(R)$ from~\cite{MRRW77}. This would
	enforce the usage of the more advanced estimate of Litsyn~\cite[Theorem 5]{litsyn1999new} and complicate analysis
	significantly. Notice that $\delta_{LP2}(R)\neq\delta_{LP1}(R)$ only for rates $R\ge 0.305$. 
	If we focus attention only on rates where new bound is better than Blinovsky's, 
	such a strengthening only affects the case of $L=3$ and results in a rather minuscule improvement (for example,
	for rate $R=0.33$ the improvement is $\approx 3\cdot 10^{-5}$). 

	\item For even $L$ it appears that $h(\beta)=R$ is no longer optimal. However, the resulting bound does not
	appear to improve upon Blinovsky's.

	\item When $L$ is large (e.g. $35$) the maximum in~\eqref{eq:thmain} is not always attained by either $j=1$ or
	$\xi_0=\delta_{LP1}(R)$. It is not clear whether such anomalies only happen in the region of rates where our
	bound is inferior to Blinovsky's.

	\item The result of Corollary~\ref{th:slope} follows by weakening~\eqref{eq:thmain} (via concavity of $g_j$, 
	Lemma~\ref{th:f1}) to
			\begin{equation}\label{eq:slo1}
				\limsup_{n\to\infty} \tau_{L}(\matc_n) \le \max_{j,\xi_0} g_{j}\left(\xi_0\right) = \max_j
			g_j(\delta_{LP1}(R))\,. 
			\end{equation}			
		The $R<R_0(L)$ condition is only used to show that the maximum is attained at $j=1$. Note also that
		weakening~\eqref{eq:slo1} corresponds to omitting the extra Elias-Bassalygo type reduction, which is
		responsible for the extra optimization over $\xi_1$ in~\eqref{eq:thmain}.
\end{enumerate}

Finally, at the invitation of anonymous reviewer we give our intuition about why our bound outperforms Blinovsky's for
odd $L$. It is easiest to compare with the weakening~\eqref{eq:slo1} of our bound. 
Now compare the two proofs:
\begin{enumerate} 
	\item Blinovsky~\cite{blinovsky1986bounds} first uses Elias-Bassalygo reduction to restrict attention to a
	subcode $\matc'$ situated on a Hamming sphere of radius $\approx \delta_{GV}(R)=h^{-1}(1-R)$. Then he proves an
	upper bound for $\tau_L(\matc')$ valid as long as $|\matc'|\gg 1$ via a Plotkin-type argument together with a
	great symmetrization idea.
	\item Our bound (following Ashikhmin, Barg and Litsyn~\cite{ashikhmin2000new}) instead uses a Kalai-Linial~\cite{KL95} reduction to select a subcode $\matc''$ situated on a Hamming
	sphere of radius $\approx \delta_{LP1}(R)$. We then proceeded to prove a (Plotkin-type) upper bound 
	on a strange quantity:
	$$ \tau_L^o(\matc'') = {1\over n} \left( \min\left\{ \rad(\{0\} \cup S): S\in {\matc \choose L}\right\} - 1
		\right)\,, $$
	which corresponds to a requirement that the code contain not more than $L-1$ codewords in any ball of radius
	$\tau_L^o$, but only for those balls that happen to also contain the origin.
\end{enumerate}
Notice that the sphere returned by Kalai-Linial is bigger than that of Elias-Bassalygo (which is the reason our bound
deteriorates at large rates), but the good thing is that the subcode $\matc''$ has another codeword $c_0$ at the center
of the Hamming sphere. Now, intuitively $\tau_L^o$ is roughly equivalent to $\tau_{L-1}$. The zero-rate (Plotkin) radius
for a list-$L$ decoding of binary codes on Hamming sphere $S_{\xi n}^n$ is given by
$$ p_L(\xi) = {\EE[\min(W_\xi, L+1-W_\xi)]\over L+1} \,, W_\xi \sim \mathrm{Bino}(L+1,\xi)\,.$$
So intuitively, we expect that Blinovsky's bound should give
	$$ \tau_L^*(R) \simleq p_{L}(\delta_{GV}(R)) $$
while our bound should give
	$$ \tau_L^*(R) \simleq p_{L-1}(\delta_{LP1}(R))\,.$$
Finally, it is easy to check that for even $L$ we have $p_{L}=p_{L-1}$, while for odd $L$, $p_L>p_{L-1}$. This is the
main intuitive reason why our bound succeeds in improving Blinovsky's, but only for odd $L$.

\section{Proofs}
\subsection{Proof of Theorem~\ref{th:main}}

Consider an arbitrary sequence of codes $\matc_n$ of rate $R$. As in~\cite{ashikhmin2000new} we start by using Delsarte's linear
programming to select a large component of the
distance distribution of the code. Namely, we apply result of Kalai and Linial~\cite[Proposition 3.2]{KL95}: 
For every $\beta$ with $h(\beta)\le R$
there exists
a sequence $\epsilon_n \to 0$ such that for every code $\matc$ of rate $R$ there is a $\xi_0$
satisfying~\eqref{eq:xi0ineq}  such that
\begin{multline}\label{eq:lp}
	A_{\xi_0 n}(\matc) \eqdef {1\over |\matc|} \sum_{x,x'\in \matc} 1\{|x-x'|=\xi_0 n\} \\
		\ge \exp\{n (R+h(\beta) - 2E_\beta(\xi_0) + \epsilon_n)\}\,.
\end{multline}

Without loss of generality (by compactness of the interval $[0,1/2-\sqrt{\beta (1-\beta)}]$ and passing to a proper
subsequence of codes $\matc_{n_k}$) we may assume that
$\xi_0$ selected in~\eqref{eq:lp} is the same for all blocklengths $n$. Then
there is a sequence of subcodes $\matc'_n$ of asymptotic rate 
	$$ R' \ge R+h(\beta) - 2E_\beta(\xi_0) $$
such that each $\matc'_n$ is situated on a sphere $c_0 + S_{\xi_0}$ surrounding another codeword $c_0 \in \matc$. Our key geometric result is: If there are too many codewords on a sphere $c_0+S_{\xi_0}$ then it is
possible to find $L$ of them that are includable in a small ball that also contains $c_0$. Precisely, we have:

\begin{lemma}\label{th:john} 
Fix $\xi_0 \in (0,1)$ and positive integer $L$. There exist a sequence $\epsilon_n\to0$ such that
for any code $\matc'_n \subset S_{\xi_0 n}$ of rate $R'>0$ there exist $L$ codewords $c_1, \ldots, c_L \in \matc'_n$ such that
\begin{equation}\label{eq:john}
	{1\over n} \rad(0, c_1, \ldots, c_L) \le \theta(\xi_0, R', L) + \epsilon_n\,,
\end{equation}
where 
\begin{align}\label{eq:thetadef} \theta(\xi_0, R', L) &\eqdef \max_{j} \theta_j(\xi_0, R', L)\\
		\theta_j(\xi_0, R', L) &\eqdef \xi_0 g_{j}\left(1-{\xi_1\over 2\xi_0}\right) 
			+ (1-\xi_0) g_j\left(\xi_1\over 2(1-\xi_0)\right)\,, \label{eq:thetajdef}
\end{align}			
	with $\xi_1=\xi_1(\xi_0)$ found as unique solution on interval $[0, 2\xi_0(1-\xi_0)]$ of
\begin{equation}\label{eq:xi1def}
	R' = h(\xi_0) - \xi_0 h\left(\xi_1\over 2\xi_0\right) - (1-\xi_0) h\left(\xi_1\over 2(1-\xi_0)\right)\,,
\end{equation}	
functions $g_j$ are defined in~\eqref{eq:gdef} and $j$ in maximization~\eqref{eq:thetadef} ranging over the same set as
in Theorem~\ref{th:main}.
\end{lemma}

Equipped with Lemma~\ref{th:john} we immediately conclude that 
\begin{equation}\label{eq:ed_3}
	\limsup_{n\to\infty} \tau_{L}(\matc_n) \le \max_{\xi_0\in[0,\delta]} \theta(\xi_0, R+h(\beta)-2E_\beta(\xi_0), L)\,.
\end{equation}
Clearly,~\eqref{eq:ed_3} coincides with~\eqref{eq:thmain}. So it suffices to prove Lemma~\ref{th:john}.

\subsection{Proof of Lemma~\ref{th:john}}

Let $\matt_L$ be the $(2^L-1)$-dimensional space of probability distributions on $\FF_2^L$. If $T\in \matt_L$ then we have
$$ T=(t_{v}, v\in\FF_2^L) \qquad t_v \ge 0, \sum_{v} t_v = 1\,.$$
We define distance on $\matt_L$ to be the $L_\infty$ one:
$$ \|T - T'\|  \eqdef \max_{v\in\FF_2^L} |t_v - t'_v|\,.$$
Permutation group $S_L$ acts naturally on $\FF_2^L$ and this action descends to
probability distributions $\matt_L$. 
%
We will say that $T$ is symmetric if
$$ T=\sigma(T) \quad \iff \quad t_v = t_{\sigma(v)} \quad \forall v \in \FF_2^L $$
for any permutation $\sigma: [L]\to[L]$. 
Note that symmetric $T$ is completely specified by $L+1$ numbers (weights of Hamming spheres in $\FF_2^L$):
	$$\sum_{v: |v|=j} t_v\,, \qquad j = 0, \ldots, L\,. $$

Next, fix some total ordering of $\FF_2^n$ (for example, lexicographic). Given a subset $S\subset \FF_2^n$ we will say
that $S$ is given in ordered form if $S=\{x_1, \ldots, x_{|S|}\}$ and $x_1<x_2 \cdots < x_{|S|}$ under the fixed ordering on $\FF_2^n$.
For any subset of codewords $S=\{x_1,
\ldots, x_L\}$ given in ordered form we define its \textit{joint type
$T(S)$} as an element of $\matt_L$ with
$$ t_v \eqdef {1\over n} |\{j: x_1(j) = v_1, \ldots, x_L(j) = v_j\}|\,,$$
where here and below $y(j)$ denotes the $j$-th coordinate of binary vector $y\in\FF_2^n$. In this way every subset $S$
is associated to an element of $\matt_L$. Note that $T(S)$ is symmetric if and only if the $L\times n$ binary matrix
representing $S$ (by combining row-vectors $x_j$) has the property that the number of columns equal to $[1, 0,\ldots,
0]^T$ is the same as the number of columns $[0, 1, \ldots, 0]^T$ etc. 
%
For any code $\matc \subset \FF_2^n$ we define its average joint type:
$$ \bar T_L(\matc) = {1\over L! \cdot {|\matc| \choose L}} \sum_\sigma \sum_{ S\in {\matc\choose L}} \sigma(T(S))\,.$$
Evidently, $\bar T_L(\matc)$ is symmetric.

Our proof crucially depends on a (slight extension of the) brilliant idea of Blinovsky~\cite{blinovsky1986bounds}:

\begin{lemma}\label{th:j0} For every $L\ge1$, $K\ge L$ and $\delta>0$ there exist a constant $K_1=K_1(L,K,\delta)$ such that for all
$n\ge 1$ and all codes
$\matc\subset \FF_2^n$ of size $|\matc|\ge K_1$ there exists a subcode $\matc' \subset \matc$ of size at least $K$ such that for any $S \in {\matc' \choose L}$ we have
\begin{equation}\label{eq:j0}
	\|T(S) - \bar T_L(\matc')\| \le \delta\,. 
\end{equation}
\end{lemma}
\begin{remark} Note that if $S' \subset S$ then every element of $T(S')$ is a sum of $\le 2^L$ elements of $T(S)$.
Hence, joint types $T(S')$ are approximately symmetric also for smaller subsets $|S'| < L$.
\end{remark}
\begin{proof} 
We first will show that for any $\delta_1>0$ and sufficiently large $|\matc|$ we may select a subcode $\matc'$ so that the following holds:
For any pair of subsets $S, S' \subset \matc'$ s.t. $|S|=|S'| \le L$ we have:
	\begin{equation}\label{eq:j1}
		\|T(S) - T(S')\| \le \delta_1
\end{equation}	

Consider any code $\matc_1 \subset \FF_2^n$ and define a hypergraph with vertices indexed by elements of
$\matc$ and hyper-edges corresponding to each of the subsets of size $L$. Now define a $\delta_1/2$-net on
the space $\matt_L$ and label each edge according to the closest element of the $\delta_1/2$-net. By a theorem of Ramsey
there exists $K_L$ such that if $|\matc_1|\ge K_L$ then there is a subset $\matc_1'\subset \matc$ such that $|\matc'_1|\ge
K$ and each of the internal edges, indexed by ${\matc_1' \choose L}$, is assigned the same label. Thus, by triangle
inequality~\eqref{eq:j1} follows for all $S, S'\in {\matc_1' \choose L}$.

Next, apply the previous argument to show that there is a constant $K_{L-1}$ such that 
for any $\matc_2 \subset \FF_2^n$ of size $|\matc_2| \ge K_{L-1}$ there exists a subcode $\matc_2'$ of size
$|\matc_2'|\ge K_L$ satisfying~\eqref{eq:j1} for all $S, S' \in {\matc_2' \choose L-1}$. Since $\matc_2'$ satisfies
the size assumption on $\matc_1$ made in previous paragraph, we can select a further subcode $\matc_2'' \subset
\matc_2'$ of size $\ge K_L$ so that for $\matc_2''$ property~\eqref{eq:j1} holds for all $S,S'$ of size $L$ or $L-1$. 

Continuing similarly, we may select a subcode $\matc'$ of arbitrary $\matc$ such that~\eqref{eq:j1} holds for all
$|S|=|S'|\le L$ provided that $|\matc|\ge K_1$.

Next, we show that~\eqref{eq:j1} implies
\begin{equation}\label{eq:j2}
	\|T(S_0) - \sigma(T(S_0))\| \le C \delta_1\,, 
\end{equation}
where $S_0 \in {\matc'\choose L}$ is arbitrary and $C=C(L)$ is a constant depending on $L$ only. \apxonly{Another way to
prove symmetry was used in~\cite[Lemma 1]{blinovsky2005code}.}

Now to prove~\eqref{eq:j2} let $T(S_0) = \{t_v, v\in \FF_2^L\}$ and consider an arbitrary transposition
$\sigma:[L]\to[L]$. It will be clear that our proof does not depend on what transposition is chosen, so for simplicity we 
take $\sigma = \{(L-1) \leftrightarrow L\}$. We want to show that~\eqref{eq:j1} implies
\begin{equation}\label{eq:j3}
	|t_v - t_{\sigma(v)}| \le \delta_1\,. \qquad \forall v\in\FF_2^L
\end{equation}
Since transpositions generate permutation group $S_L$,~\eqref{eq:j2} then follows. Notice that~\eqref{eq:j3} is only
informative for $v$ whose last two digits are not equal, say $v=[v_0, 0, 1]$. 
Suppose that $S_0 = \{c_1, \ldots, c_L\}$ given in the ordered form. Let
\begin{align} S &= \{c_1, \ldots c_{L-1}\}\,,\\
   S' &= \{c_1, \ldots, c_{L-2}, c_L\}
\end{align}
Joint types $T(S)$ and $T(S')$ are expressible as functions of $T(S_0)$ in particular, the number of 
occurrences of element $[v_0, 0]$ in $S$ is
$ t_{[v_0, 0, 1]} + t_{[v_0, 0, 0]} $
and in $S'$ is $t_{[v_0, 0, 0]} + t_{[v_0, 1, 0]}$. Thus, from~\eqref{eq:j1} we obtain:
$$ | (t_{[v_0, 0, 1]} + t_{[v_0, 0, 0]}) - (t_{[v_0, 0, 0]} + t_{[v_0, 1, 0]})| \le \delta $$
implying~\eqref{eq:j3} and thus~\eqref{eq:j2}.

Finally, we show that~\eqref{eq:j2} implies~\eqref{eq:j0}. Indeed, consider the chain
\begin{align} \lefteqn{\|T(S) - \bar T_L(\matc')\|} \nonumber\\
{} &= \left\|T(S) - {1\over L! \cdot {|\matc'| \choose L}} \sum_\sigma \sum_{ S'\in
{\matc'\choose L}} \sigma(T(S'))\right\| \\
		&\le {1\over L! \cdot {|\matc'| \choose L}} \sum_\sigma \sum_{ S'\in
{\matc'\choose L}}\|T(S) -  \sigma(T(S'))\| \label{eq:j3a}\\
		&\le {1\over L! \cdot {|\matc'| \choose L}} \sum_\sigma \sum_{ S'\in
{\matc'\choose L}} \|T(S) -  T(S')\| \nonumber\\
{} & {}+\|T(S') - \sigma(T(S'))\|\label{eq:j3b}\\
		&\le (1+C) \delta_1\,, \label{eq:j3c}
\end{align}		
where~\eqref{eq:j3a} is by convexity of the norm,~\eqref{eq:j3b} is by triangle inequality and~\eqref{eq:j3c} is
by~\eqref{eq:j1} and~\eqref{eq:j2}.
Consequently, setting $\delta_1 = {\delta\over 1+C}$ we have shown~\eqref{eq:j0}.
\end{proof}

\smallskip
Before proceeding further we need to define the concept of an average radius (or a moment of inertia):
$$\arad(x_1, \ldots, x_m) \eqdef \min_y {1\over m} \sum_{i=1}^m |x_i -y |\,. $$
Note that the minimizing $y$ can be computed via a per-coordinate majority vote (with arbitrary tie-breaking for even $m$).
Consider now an arbitrary subset $S=\{c_1,\ldots, c_L\}$ and define for each $j\ge 0$ the following functions
$$ h_j(S) \eqdef {1\over n} \arad(\underbrace{0,\ldots,0}_{j\mbox{~times}}, c_1, \ldots, c_L)\,. $$
It is easy to find an expression for $h_j(S)$ in terms of the joint-type of $S$:
\begin{align}\label{eq:g1}
	h_j(S) &= {1\over L+j} \left( \EE[W] - \EE[|2W-L-j|^+] \right)\\
	\PP[W=w] &= \sum_{v: |v|=w} t_v\,,
\end{align}
where $t_v$ are components of the joint-type $T(S) = \{t_v, v\in \FF_2^L\}$. To check~\eqref{eq:g1} simply observe that
if one arranges $L$ codewords of $S$ in an $L\times n$ matrix and also adds $j$ rows of zeros, then computation of
$h_j(S)$ can be done per-column: each column of weight $w$ contributes
$$ \min(w, L+j-w) = w-|2w-L-j|^+ $$
to the sum. In view of expression~\eqref{eq:g1} we will abuse notation and write
$$ h_j(T(S)) \eqdef h_j(S)\,.$$

We now observe that for symmetric codes satisfying~\eqref{eq:j0}  average-radii $h_j(S)$ in fact determine the regular
radius:
\begin{lemma}\label{th:g2} Consider an arbitrary code $\matc$ satisfying conclusion~\eqref{eq:j0} of Lemma~\ref{th:j0}.
Then for any subset  $S = \{c_1,\ldots,c_L\} \subset \matc$ we have
\begin{equation}\label{eq:g2}
	\left| \rad(0, c_1, \ldots, c_L) - n \cdot \max_j h_j(\bar T_L(\matc))\right| \le 2^L (1+\delta n)\,,
\end{equation}
where $j$ in maximization~\eqref{eq:g2} ranges over $\{0, 1, 3, \ldots, 2k+1, \ldots, L\}$ if $L$ is odd and over $\{0,
	2, \ldots,2k, \ldots L\}$ if $L$ is even.
\end{lemma}
\begin{proof} For joint-types of size $L$ and all $j\ge 0$ we clearly have
	(cf. expression~\eqref{eq:g1})
	\begin{equation}\label{eq:g1a}
		|h_j(T_1) - h_j(T_2)| \le 2^{L-1} \|T_1 - T_2\|\,, \qquad \forall T_1,T_2 \in \matt_L\,. 
\end{equation}	
	We also trivially have
	\begin{equation}\label{eq:g1b}
		{1\over n} \rad(0, c_1, \ldots, c_L) \ge h_j (S) \qquad \forall j\ge 0\,.
\end{equation}	
	Thus from~\eqref{eq:j0} and~\eqref{eq:g1a} we already get
		$$ {1\over n}\rad(0, c_1, \ldots, c_L) \ge \max_j h_j(\bar T_L(\matc)) - 2^{L-1} \delta\,.$$
It remains to show
\begin{equation}\label{eq:g3}
	{1\over n}\rad(0, c_1, \ldots, c_L) \le  \max_j h_j(\bar T_L(\matc)) + \delta + {2^L\over n}\,.
\end{equation}
	This evidently requires constructing a good center $y$ for the set $\{0, c_1,\ldots,c_L\}$.
	To that end fix arbitrary numbers $q=(q_0, \ldots, q_L) \in [0,1]^L$. Next, for each
	$v\in\FF_2^L$ let $E_v \subset [n]$ be all coordinates on which restriction of $\{c_1,\ldots, c_L\}$ equals
	$v$. On $E_v$ put $y$ to have a fraction $q_{|v|}$ of ones and remaining set to zeros (rounding to integers arbitrarily). Proceed for all $v\in
	\FF_2^L$. Call resulting vector $y(q)\in \FF_2^n$. 

	Denote for convenience $c_0=0$. We clearly have
	\begin{equation}\label{eq:g4}
		\rad(c_0, c_1, \ldots, c_L) \le \min_q \max_{p} \sum_{i=0}^L p_i |c_i - y(q)|\,,
\end{equation}	
	where $p=(p_0, \ldots, p_L)$ is a probability distribution.  

	Denote 
		\begin{align} T(S)	&= \{t_v, v\in\FF_2^L\}\\
		   \bar T_L(\matc) &= \{\bar t_v, v\in\FF_2^L\} 
	\end{align}
We proceed to computing $|c_i - y(q)|$.
	\begin{multline} 
		|c_i - y(q)| \le n\sum_{v\in\FF_2^L} t_v (q_{|v|} 1\{v(i) = 0\} \\
			+ (1-q_{|v|}) 1\{v(i) = 1\}) + 2^L\,,
	\end{multline}	
	where $2^L$ comes upper-bounding the integer rounding issues and we abuse notation slightly by setting $v(0)=0$
	for all $v$ (recall that $v(i)$ is the $i$-th coordinate of $v\in\FF_2^L$).

	By~\eqref{eq:j0} we may replace $t_v$ with $\bar t_v$ at the expense of introducing $2^L \delta n$ error, so we
	have:
		\begin{multline} |c_i - y(q)| \le n\sum_{v\in\FF_2^L} \bar t_v (q_{|v|} 1\{v(i) = 0\} \\
			+ (1-q_{|v|}) 1\{v(i) = 1\}) + 2^L(1+\delta n)\,.
	\end{multline}		
	Next notice that the sum over $v$ only depends on whether $i=0$ or $i\neq 0$ (by symmetry of $\bar t_v$).
	Furthermore, for any given weight $w$ and $i\neq 0$ we have
		$$ \sum_{v: |v|=w} 1\{v(i)=1\} = {L \choose w}{w\over L}\,.$$
	Thus, introducing the random variable $\bar W$, cf.~\eqref{eq:g1},
	$$ \PP[\bar W = w] \eqdef \sum_{v:|v|=w} \bar t_v\,,$$
	we can rewrite:
	\begin{multline} \sum_{v\in\FF_2^L} \bar t_v (q_{|v|} 1\{v(i) = 0\} + (1-q_{|v|}) 1\{v(i) = 1\}) \\= 
		{1\over L} \EE[\bar W + (L-2\bar W) q_{\bar W}]\,. 
	\end{multline}		
	For $i=0$ the expression is even simpler:
	$$ \sum_{v\in\FF_2^L} \bar t_v (q_{|v|} 1\{v(0) = 0\} + (1-q_{|v|}) 1\{v(0) = 1\}) =
		\EE[q_{\bar W}]\,.$$
	
	Substituting derived upper bound on $|c_i-y(q)|$ into~\eqref{eq:g4} we can see that without loss of generality
	we may assume $p_1=\cdots=p_L$, so our upper bound (modulo $O(\delta)$ terms) becomes:
	\begin{align*} \lefteqn{\min_q \max_{p_1 \in [0, L^{-1}]} (1-Lp_1) \EE[q_{\bar W}] + p_1 \EE[\bar W + (L-2\bar W) q_{\bar
	W}] }\\
		& = \min_q \max_{p_1 \in [0, L^{-1}]} p_1 \EE[\bar W] + \EE[q_{\bar W}(1-2\bar W p_1)] 
\end{align*}		
	By von Neumann's minimax theorem we may interchange min and max, thus continuing as follows:
	\begin{align}
		& = \max_{p_1 \in [0, L^{-1}]} \min_q  p_1 \EE[\bar W] + \EE[q_{\bar W}(1-2\bar W p_1)] \\
		&= \max_{p_1 \in [0, L^{-1}]} p_1 \EE[\bar W] - \EE[|2 \bar W p_1 -1|^+]\,.\label{eq:g5}
	\end{align}
	The optimized function of $p_1$ is piecewise-linear, so optimization can be reduced to comparing values at
	slope-discontinuities and boundaries. The point $p_1 = 0$ is easily excluded, while the rest of the points are
	given by $p_1 = {1\over L+j}$ with $j$ ranging over the set specified in the statement of Lemma\footnote{The
	difference
	between odd and even $L$ occurs due to the boundary point $p_1 = {1\over L}$ not being a
	slope-discontinuity when $L$ is odd, so we needed to add it separately.}. So we continue~\eqref{eq:g5} getting
	\begin{align}
		&= \max_{j} {1\over L+j} \left(\EE[\bar W] - \EE[|2 \bar W - L - j|^+]\right)\label{eq:g6}
	\end{align}
	We can see that expression under maximization is exactly $h_j(\bar T_L(\matc))$ and hence~\eqref{eq:g3} is
	proved.
\end{proof}

\apxonly{A similar method shows that radius and weighted average radius are related:
\begin{align} \rad(c_1, \ldots, c_L) &\le \max_p \arad\left(\begin{matrix}c_1 & \ldots & c_L\\ p_1 & \ldots &
p_L\end{matrix}\right) + 2^L\\
\rad(c_1, \ldots, c_L) &\ge \max_p \arad\left(\begin{matrix}c_1 & \ldots & c_L\\ p_1 & \ldots &
p_L\end{matrix}\right)
\end{align}
}

\begin{lemma}\label{th:e1} There exist constants $C_1, C_2$ depending only on $L$ such that for any $\matc \subset
\FF_2^n$  the joint-type $\bar T_L(\matc)$ is approximately
	a mixture of product Bernoulli distributions\footnote{Distribution $\mathrm{Bern}^{\otimes L}(\lambda)$ assigns
	probability $\lambda^{|v|}(1-\lambda)^{L-|v|}$ to element $v\in\FF_2^L$.}, namely:
	\begin{equation}\label{eq:e1a}
		\left\|\bar T_L(\matc) - {1\over n} \sum_{i=1}^n \mathrm{Bern}^{\otimes L}(\lambda_i)\right\| \le {C_1\over
	|\matc|}\,, 
\end{equation}	
	where $\lambda_i = {1\over |\matc|} 
	\sum_{c\in\matc} 1\{c(i) =1\}$ be the density of ones in the $j$-th column of a $|\matc|\times n$ matrix
	representing the code. 
	 In particular, 
	\begin{equation}\label{eq:e1b}
		\left|h_j(\bar T_L(\matc)) - {1\over n} \sum_j g_j(\lambda_j) \right| \le {C_2\over
		|\matc|}\,,
\end{equation}	
	where functions $g_j$ were defined in~\eqref{eq:gdef}.
\end{lemma}
\begin{proof} 
Second statement~\eqref{eq:e1b} follows from the first via~\eqref{eq:g1a} and linearity of $h_j(T)$ in the type $T$,
cf.~\eqref{eq:g1}. To show the first statement, let $M=|\matc|$, $M_i = \lambda_i M$ and $p_w$ -- total probability 
assigned to vectors $v$ of weight $w$ by $\bar T_L(\matc)$. Then by computing $p_w$ over columns of $M\times n$ matrix
we obtain
	$$ p_w = {1\over n} \sum_{i=1}^n {{M_i \choose w}  {M-M_i \choose L-w} \over {M\choose L}}\,.$$
By a standard estimate we have for all $w=\{0,\ldots,L\}$:
$$ {{M_i \choose w}  {M-M_i \choose L-w} \over {M\choose L}} = {L\choose w} \lambda_i^w (1-\lambda_i)^{L-w} + O({1\over
M})\,,$$
with $O(\cdot)$ term uniform in $w$ and $\lambda_i$. By symmetry of the type $\bar T_L(\matc)$ the result~\eqref{eq:e1a}
follows.
\end{proof}

\begin{lemma}\label{th:f1} Functions $g_j$ defined in~\eqref{eq:gdef} are concave on $[0,1]$.
\end{lemma}
\begin{proof}
Let $W_\lambda\sim\mathrm{Bino}(L, \lambda)$ and $V_\lambda \sim \mathrm{Bino}(L-1, \lambda)$. Denote for convenience
$\bar \lambda = 1-\lambda$ and take $j_0$ to be an integer between $0$ and $L$. We have then
\begin{align} \lefteqn{{\partial\over \partial \lambda} \EE[|W_\lambda-j_0|^+] }\nonumber\\
	&= \sum_{w=j_0+1}^L {L \choose w}(w-j_0) \lambda^{w} \bar\lambda^{L-w} \left\{w \lambda^{-1} -(L-w) 
	\bar\lambda^{-1}\right\}\label{eq:f2a}\\
	&= {L\choose j_0 + 1}(j_0+1) \lambda^{j_0} \bar\lambda^{L-j_0 - 1} \nonumber\\
	&{} + \sum_{w=j_0+1}^{L-1} \bigg[{L\choose w+1} (w+1-j_0) (w+1){}\nonumber\\
	&{}  
			 - {L\choose w}(w-j_0)(L-w)\bigg] \lambda^w
	\bar\lambda^{L-w-1}\label{eq:f2} \\
	&= L{L-1\choose j_0} \lambda^{j_0} \bar\lambda^{L-1-j_0} + L \sum_{w=j_0+1}^{L-1} {L-1\choose w} \lambda^w
	\bar\lambda^{L-1-w}\label{eq:f3}\\
	&= L \PP[V_\lambda \ge j_0]\,,\label{eq:f4}
\end{align}
where in~\eqref{eq:f2} we shifted the summation by one for the first term under the sum in~\eqref{eq:f2a}, and
in~\eqref{eq:f3} applied identities ${L\choose w+1}={L\choose w}{L-w\over w+1}={L-1\choose w}{L\over
w+1}$. Similarly, if $\theta\in[0,1)$ we have
\begin{equation}\label{eq:f5} 
{\partial\over \partial \lambda} \EE[|W_\lambda-j_0-\theta|^+] = L\PP[V_\lambda \ge j_0 + 1] + L(1-\theta) \PP[V_\lambda
= j_0]\,.
\end{equation}
Similarly, one shows (we will need it later in Lemma~\ref{th:n}):
\begin{equation}\label{eq:f5a}
	{\partial\over \partial\lambda} \PP[W_\lambda \ge j_0] = L \PP[V_\lambda = j_0 - 1]\,.
\end{equation}

Since clearly the function in~\eqref{eq:f5} is strictly increasing in $\lambda$ for any $j_0$ and $\theta$ we conclude that 
$$ \lambda \mapsto \EE[|W_\lambda-j_0-\theta|^+] $$
is convex. This concludes the proof of concavity of $g_j$.
\end{proof}

\medskip
\begin{proof}[Proof of Lemma~\ref{th:john}] Our plan is the following:
\begin{enumerate}
\item Apply Elias-Bassalygo reduction to pass from $\matc'_n$ to a subcode $\matc''_n$ on an intersection of two spheres $S_{\xi_0
n}$ and $y+S_{\xi_1 n}$.
\item Use Lemma~\ref{th:j0} to pass to a symmetric subcode $\matc'''_n \subset \matc''_n$
\item Use Lemmas~\ref{th:e1}-\ref{th:f1} to estimate maxima of average radii $h_j$ over $\matc'''_n$.
\item Use Lemma~\ref{th:g2} to transport statement about $h_j$ to a statement on $\tau_L(\matc'''_n)$.
\end{enumerate}

We proceed to details. 
It is sufficient to show that for some constant $C=C(L)$ and arbitrary $\delta>0$ estimate~\eqref{eq:john} holds with
$\epsilon_n = C \delta$ whenever $n\ge n_0(\delta)$. So we fix $\delta>0$ and
consider a code $\matc' \subset S_{\xi_0 n} \subset \FF_2^n$ with $|\matc'|\ge \exp\{nR'+o(n)\}$. Note that for any $r$
, even
$m$ with $m/2 \le \min(r, n-r)$ and arbitrary $y\in S_r^n$ intersection $\{y+S_m^n\} \cap S_r^n$ is isometric to the
product of two lower-dimensional spheres:
\begin{equation}\label{eq:jx}
	\{y+S_m^n\} \cap S_r^n  \cong S_{r-m/2}^r \times S_{m/2}^{n-r}\,. 
\end{equation}
Therefore, we have for $r=\xi_0 n $ and valid $m$:
$$ \sum_{y\in S_r^n} |\{y + S_{m}^n\} \cap \matc'| = |\matc'| {\xi_0 n \choose \xi_0n - m/2} {n(1-\xi_0) \choose m/2}\,.$$
Consequently, we can select $m=\xi_1 n - o(n)$, where $\xi_1$ defined in~\eqref{eq:xi1def}, so that for some
$y\in S_r^n$:
$$ |\{y + S_{\rho n}^n\} \cap \matc'| > n\,. $$
Note that we focus on solution of~\eqref{eq:xi1def} satisfying $\xi_1 < 2\xi_0(1-\xi_0)$. For some choices of $R,
\delta$ and $\xi_0$ choosing $\xi_1 >
2\xi_0(1-\xi_0)$ is also possible, but such a choice appears to result in a weaker bound.

Next, we let $\matc'' =\{y + S_{\rho n}^n\} \cap \matc'$.
For sufficiently large $n$ the code $\matc''$ will satisfy assumptions of Lemma~\ref{th:j0} with $K\ge {1\over
\delta}$. Denote the resulting large symmetric subcode $\matc'''$. 

Note that because of~\eqref{eq:jx} column-densities $\lambda_i$'s of $\matc'''$, defined in Lemma~\ref{th:e1}, satisfy (after possibly
reordering coordinates):
$$ \sum_{i=1}^{\xi_0 n} \lambda_i = \xi_1 n/2 + o(n), \quad \sum_{i>\xi_0 n} \lambda_i = \xi_1 n/2 + o(n)\,. $$
Therefore, from Lemmas~\ref{th:e1}-\ref{th:f1} we have
\begin{multline}\label{eq:jx2}
	h_j(\bar T_L(\matc''')) \le \xi_0 g_j\left(1-{\xi_1\over 2\xi_0}\right) \\
		+ (1-\xi_0) g_j\left(\xi_1\over 2(1-\xi_0)\right) + \epsilon'_n + {C_1\over |\matc'''|}\,,
\end{multline}	
where $\epsilon'_n\to0$. Note that by construction the last
term in~\eqref{eq:jx2} is $O(\delta)$. Also note that the first two terms in~\eqref{eq:jx2} equal $\theta_j$ defined
in~\eqref{eq:thetadef}.

Finally, by Lemma~\ref{th:g2} we get that for any codewords $c_1, \ldots, c_L \in \matc'''$, some constant $C$ and
some sequence $\epsilon''_n\to0$ the following holds:
$$ {1\over n} \rad(0, c_1,\ldots, c_L) \le \theta(\xi_0, R', L) + \epsilon''_n + C\delta\,. $$
By the initial remark, this concludes the proof of Lemma~\ref{th:john}.
\end{proof}

\subsection{Proof of Corollary~\ref{th:slope}}\label{sec:slope}

\begin{lemma}\label{th:n} For any odd $L=2a+1$ there exists a neighborhood of $x={1\over2}$ such that 
\begin{equation}\label{eq:n1a}
	\max_j g_j(x) = g_1(x)\,,
\end{equation}	
	maximum taken over $j$ equal all the odd numbers not exceeding $L$ and $j=0$.
	We also have for some $c>0$
	\begin{equation}\label{eq:n2}
		g_1(x) = {1\over 2} - 2^{-L-1}{L\choose {L-1\over2}} + cx + O((2x-1)^2), \qquad x\to{1\over2}\,.
	\end{equation}	
\end{lemma}
\begin{proof} First, the value $g_1(1/2)$ is computed trivially. Then from~\eqref{eq:f5} we have
	\begin{equation}\label{eq:nx1}
		{d\over dx} g_j(x) = {L\over L+j}\left(1 - 2\PP\left[V_x \ge {L+j\over2}\right]\right),
\end{equation}	
	where $j\ge 1$ and $V_x \sim \mathrm{Bino}(x, L-1)$. This implies~\eqref{eq:n2}. For future reference we note that~\eqref{eq:n1b} (below) and~\eqref{eq:f5a}
	imply
	\begin{multline}\label{eq:nx2}
		{d\over dx} g_0(x) = 1- 2\PP[V_x\ge {L+1\over2}] - \PP[V_x={L-1\over2}], \\
			\quad V_x \sim \mathrm{Bino}(x,	L-1)\,.  
	\end{multline}	
\apxonly{Also: $$ {d\over dx} g_j(x) = -{2L(L-1)\over L+j} \PP\left[\mathrm{Bino}(x, L-2) = {L+j-2\over2}\right] $$}
By continuity,~\eqref{eq:n1a} follows from showing
\begin{equation}\label{eq:n1}
		g_1(1/2) > \max_{j\in\{0, 3, 5, \ldots L\}} g_j(1/2)\,.
\end{equation}	
	Next, consider $W_x\sim \mathrm{Bino}(x, L)$ and notice the upper-bound
	$$ g_j(x) \le {1\over L+j} \EE\left[ W_x1\{W_x\le a\} + (L+j-W_x)1\{W_x \ge a+1\} \right]\,.$$
	Then, substituting expression for $g_1(x)$ we get 
	\begin{align} g_1(x) - g_0(x) &= {1\over L}\left(\PP[W_x\ge a+1] - g_1(x)\right)\label{eq:n1b}\\
	   g_1(x) - g_j(x) &\ge {j-1\over L+j}\left(g_1(x) - \PP[W_x>a+1]\right)\,. 
\end{align}
	Thus, to show~\eqref{eq:n1} it is sufficient to prove that for $x=1/2$ we have
	\begin{equation}\label{eq:n3}
		\PP[W_{1\over2} > a+1] < g_1(1/2) < \PP[W_{1\over2}\ge a+1]\,.
\end{equation}	
The right-hand inequality is trivial since $\PP[W_{1\over2}\ge a+1] = 1/2$ while from~\eqref{eq:n2} we know $g_1(1/2)<1/2$.
	The left-hand inequality, after simple algebra, reduces to showing
	\begin{equation}\label{eq:n2a}
		\sum_{u=0}^{a-1} (2a+1-2u) {2a+1 \choose u} < (2a+1) {2a+1\choose a}\,. 
\end{equation}
	Notice, that
		$$ (n-2u){n\choose u} = n\left[ {n-1 \choose u} - {n-1\choose u-1}\right]  \forall u\ge 0$$
	and therefore
		$$ \sum_{u \le \ell} (n-2u){n\choose u} = n {n-1 \choose \ell}\,.$$
	Plugging this identity into the right-hand side of~\eqref{eq:n2a} we get
	\begin{multline} \sum_{u=0}^{a-1} (2a+1-2u) {2a+1 \choose u} = (2a+1) {2a \choose a-1} \\
		< (2a+1) {2a\choose a} < (2a+1){2a+1 \choose a}
	\end{multline}	
	completing the proof of~\eqref{eq:n2a}.
\end{proof}

\begin{proof}[Proof of Corollary~\ref{th:slope}]
	We first show that~\eqref{eq:sl0} implies~\eqref{eq:zeroslope}. 
	To that end, fix a small $\epsilon>$ so that ${1\over2}-\epsilon$ belongs to the neighborhood existence of which
	is claimed in Lemma~\ref{th:n}. Choose rate so that $\delta_{LP1}(R)=1/2-\epsilon$ and notice that this implies 
	\begin{align} R &= h(\epsilon^2 + o(\epsilon^2) )\,,\label{eq:sm1}
	\end{align}
	By Lemma~\ref{th:n}, the right-hand side of~\eqref{eq:sl0} is 
		$$ \tau^*_L(0) - \mathrm{const} \cdot \epsilon + o(\epsilon)\,,$$
		which together with~\eqref{eq:sm1} implies~\eqref{eq:zeroslope}.

	To prove~\eqref{eq:sl0} we use Theorem~\ref{th:main} with $\delta=\delta_{LP1}(R)$. Next, 
	use concavity of $g_j$'s (Lemma~\ref{th:f1}) to relax~\eqref{eq:thmain} to
	$$ 	\limsup_{n\to\infty} \tau_{L}(\matc_n) \le \max_{j, \xi_0} g_j(\xi_0)\,. $$
	From~\eqref{eq:nx1} and~\eqref{eq:nx2} it is clear that $\xi_0 \mapsto g_j(\xi_0)$ is monotonically
	increasing for all $j\ge 0$ on the interval $[0,1/2]$. Thus, we further have
	\begin{equation}\label{eq:sl1}
		\limsup_{n\to\infty} \tau_{L}(\matc_n) \le \max_{j} g_j(\delta_{LP1}(R))\,. 
\end{equation}	
	Bound~\eqref{eq:sl1} is valid for all $R\in[0,1]$ and arbitrary (odd/even $L$). However, when $R$ is small (say,
	$R<R_0$) and $L$ is odd,
	$\delta_{LP1}(R)$ belongs to the neighborhood of $1/2$ in Lemma~\ref{th:n} and thus~\eqref{eq:sl0} follows from~\eqref{eq:sl1}
	and~\eqref{eq:n1a}.
\end{proof}

\begin{remark}\label{rm:evenL} It is, perhaps, instructive to explain why Corollary~\ref{th:slope} cannot be shown for even $L$ (via
	Theorem~\ref{th:main}). For even $L$ the maximum over $j$ of $g_j(1/2 - \epsilon)$ is attained at $j=0$ and 
\begin{equation}\label{eq:sm3}
		g_0({1\over 2} - \epsilon) = \tau^*_L(0) + c \epsilon^2 + O(\epsilon^3)\,,  \epsilon\to 0
\end{equation}
Therefore, for $\delta_{LP1}(R) = {1\over2}-\epsilon$ we get from~\eqref{eq:sm3} that the right-hand side
of~\eqref{eq:sl1} evaluates to 
\begin{equation}\label{eq:sm4}
		\tau_L^*(0) - \mathrm{const} \cdot \epsilon^2 \log {1\over \epsilon}\,. 
\end{equation}	
Thus, comparing~\eqref{eq:sm4} with~\eqref{eq:sm1} we conclude that for even $L$ our bound  on $R^*_L(\tau)$ has
negative slope at zero rate. Note that Blinovsky's bound~\eqref{eq:blin} has negative slope at zero rate for both
odd and even $L$.
\end{remark}

\apxonly{Apxonly remarks:
\begin{enumerate}
	\item For $L$ -- even, optimal choice is NOT $\delta = \delta_{LP}$, but is about 10\% larger for small rates.
	This does give some improvement, but still the bound appears to always be worse than Blinovsky.
\end{enumerate}
}

\subsection{Proof of Corollary~\ref{th:l3}}

\begin{proof}
 Instead of working with parameter $\delta$ we introduce $\beta \in[0,1/2]$ such that
	$$ \delta = {1\over2} - \sqrt{\beta (1-\beta)}\,. $$
We then apply Theorem~\ref{th:main} with $h(\beta)=R$. Notice that the bound on $\xi_0$ in~\eqref{eq:xi0ineq} becomes 
$$ 0 \le \xi_0 \le \delta\,. $$
By a simple substitution $\omega=\sqrt{\beta\over 1-\beta}$ we get from~\eqref{eq:ed_2}
$$ E_\beta(\delta) = {1\over2}(\log 2 - h(\delta) + h(\beta))\,. $$
Therefore, when $\xi_0 = \delta$ we notice that 
$$ R+h(\beta)-2E_\beta(\xi_0) = R-\log 2 + h(\delta) $$
implying that defining equation for $\xi_1$, i.e.~\eqref{eq:xidef}, coincides with~\eqref{eq:bd3c}.

Next for $L=3$ we compute
\begin{align} g_0(\nu) &= \nu(1-\nu)\,,\\
   g_1(\nu) &= {3\over4}\nu - {1\over 2} \nu^3\,,\\
   g_3(\nu) &= {1\over 2}\nu\,. 
\end{align}   

Note that the right-hand side of~\eqref{eq:bd3b} is precisely equal to 
$$ \delta g_1\left(1-{\xi_1\over 2\delta}\right) + 
	(1-\delta) g_1\left({\xi_1\over 2(1-\delta)}\right)\,. $$ 
So this corollary simply states
that for $L=3$ the maximum in~\eqref{eq:thmain} is achieved at $j=1, \xi_0 = \delta$. 
Let us restate this last statement rigorously: The maximum
\begin{equation}\label{eq:m1}
	\max_{j\in\{0, 1, 3\}} \max_{\xi_0 \in \delta} \xi_0 g_{j}\left(1-{x\over 2\xi_0}\right) 
			+ (1-\xi_0) g_j\left(x\over 2(1-\xi_0)\right) 
\end{equation}			
is achieved at $j=1, \xi_0=\delta$. Here $x = x(\xi_0, \beta)$ is a solution of
\begin{multline}\label{eq:m3}
	2(h(\beta) - E_\beta(\xi_0))  \\
		= h(\xi_0) - \xi_0 h\left(x\over 2\xi_0\right) - (1-\xi_0) h\left(x\over
2(1-\xi_0)\right)\,. 
\end{multline}
For notational convenience we will denote the function under maximization in~\eqref{eq:m1} by $g_j(\xi_0, x)$.

We proceed in two steps:
\begin{itemize}
\item First, we estimate the maximum over $\xi_0$ for $j=0$ as follows:
	\begin{multline}\label{eq:m2}
		\max_{\xi_0} g_0(\xi_0, x) \le {\log 2 - R\over 4 \log 2} \cdot \left(1- {1-\delta \over {a_{max}}
		(1-a_{max})}\right) \\
			+ (1-\delta) g_0(a_{min})\,,
	\end{multline}	
	where $a_{max}, a_{min} \le {1\over2}$ are given by
		\begin{align} a_{max} &= h^{-1}(\log2-R)\,,\label{eq:m2a}\\
		   a_{min} &= h^{-1}\left(\log2- {R\over 1-\delta}\right)\label{eq:m2b}\,. 
		\end{align}		   
\item Second, we prove that for $j=1$ function
	$$ \xi_0 \mapsto g_j(\xi_0, x(\xi_0)) $$
	is monotonically increasing.
\end{itemize}

Once these two steps are shown, it is easy to verify (for example, numerically) that $g_1(\delta, x(\delta))$ exceeds
both ${1\over 2}\delta$ (term corresponding to $j=3$ in~\eqref{eq:m1}) and the right-hand side of~\eqref{eq:m2} (term
corresponding to $j=0$). Notice that this relation holds for all rates. Therefore, maximum in~\eqref{eq:m1} is indeed
attained at $j=1, \xi_0=\delta$.

One trick that will be common to both steps is the following. From the proof of Lemma~\ref{th:john} it is clear that
the estimate~\eqref{eq:john} is monotonic in $R'$. Therefore, in equation~\eqref{eq:m3} we may replace $E_\beta(\xi)$
with any upper-bound of it. We will use the well-known upper-bound, which leads to binomial estimates of spectrum
components~\cite[(46)]{litsyn1999new}:
\begin{equation}
	E_\beta(\xi_0) \le {1\over2}(\log 2 +h(\beta) - h(\xi_0))\,. 
\end{equation}
Furthermore, it can also be argued that maximum cannot be attained by $\xi_0$ so small that 
$$ h(\beta) - {1\over2}(\log 2 +h(\beta) - h(\xi_0)) < 0\,. $$
So from now on, we assume that 
$$ h^{-1}(\log2 - h(\beta)) \le \xi_0 \le \delta \,,$$
and that
$x = x(\xi_0) \le 2\xi_0(1-\xi_0)$ is determined from the equation:
\begin{equation}\label{eq:m3a}
	\log 2 - R = \xi_0 h\left(x\over 2\xi_0\right) + (1-\xi_0) h\left(x\over 2(1-\xi_0)\right)
\end{equation}
(we remind $R=h(\beta)$).

We proceed to demonstrating~\eqref{eq:m2}. For convenience, we introduce
\begin{align} a_1 &\eqdef 1-{x\over 2\xi_0}\,,\label{eq:m4a}\\
   a_2 &\eqdef {x\over 2-2\xi_0}\,. \label{eq:m4b}
\end{align}   
By constraints on $x$ it is easy to see that
$$ 0 \le a_2 \le \min(a_1, 1-a_1)\,. $$
Therefore, we have
$$ \log2-R = \xi_0 h(a_1) + (1-\xi_0) h(a_2) \ge h(a_2) $$
and thus $a_2 \le a_{max}$ defined in~\eqref{eq:m2a}. Similarly, we have
$$ \log2-R = \xi_0 h(a_1) + (1-\xi_0) h(a_2) \le \xi_0 \log 2 + (1-\xi_0) h(a_2)\,, $$
and since $\xi_0 \le \delta$ we get that $a_2 \ge a_{min}$ defined in~\eqref{eq:m2b}.

Next, notice that ${h(x)\over x(1-x)}$ is decreasing on $(0,1/2]$. Thus, we have
\begin{align} 
   h(a_1) &\ge g_0(a_1) 4\log 2 \\
   h(a_2) &\ge h(a_{max}) {g_0(a_2)\over g_0(a_{max})} \nonumber\\
   	  &= {\log2-R\over a_{max}(1-a_{max})} g_0(a_2) \eqdef c \cdot
   g_0(a_2)\,,
\end{align}
where in the last step we introduced $c>4\log2$ for convenience.
Consequently, we get
\begin{align} \lefteqn{\log 2-R}\nonumber\\
{} &= \xi_0 h(a_1) + (1-\xi_0) h(a_2) \\
	&\ge 4\log 2 \cdot \xi_0 g_0(a_1)  + (1-\xi_0)  c \cdot g_0(a_2)\\
	&= 4\log 2 \cdot g_0(\xi_0, x) + (1-\xi_0) (c-4\log2) \cdot g_0(a_2)\\
	&\ge 4\log 2 \cdot g_0(\xi_0, x) + (1-\delta) (c-4\log2) \cdot g_0(a_{min})\,.
\end{align}
Rearranging terms yield~\eqref{eq:m2}.

We proceed to proving monotonicity of~\eqref{eq:m3}. The technique we will use is general (can be applied to $L>3$ and
$j>1$), so we will avoid particulars of $L=3, j=1$ case until the final step.

Notice that regardless of the function $g(\nu)$ we have the equivalence:
\begin{multline} {d\over d\xi_0}  \xi_0 g(a_1) + (1-\xi_0) g(a_2) \ge 0 \quad \iff \quad\\
		{1\over 2} {dx\over d\xi_0} (g'(a_2)-g'(a_1)) \ge \int\limits_{a_2}^{a_1} (1-x)(-g''(x)) dx - g'(a_2)\,,
		\label{eq:m4}
\end{multline}
where we recall definition of $a_1,a_2$ in~\eqref{eq:m4a}-\eqref{eq:m4b}. 
Differentiating~\eqref{eq:m3a} in $\xi_0$ (and recalling that $R$ is fixed, while $x=x(\xi_0)$ is an implicit function
of $\xi_0$) we find
$$ {dx\over d\xi_0} = -2{\log {1-a_2\over a_1} \over \log {1-a_2\over a_2} {a_1\over 1-a_1}} <0\,. $$
\apxonly{Without binomial bound we have:
$$ {dx\over d\xi_0} = (\log {1-\xi_0\over \xi_0} + 2\log{1-\omega\over 1+\omega} - F_1)/F_2 $$
where $F_1 = h(a_1)-h(a_2) + (1-a_1) \log{1-a_1\over a_1} + a_2 \log (1-a_2\over a_2)$, $F_2 = {1\over 2} ( \log
{1-a_2\over a_2} - \log {1-a_1\over a_1})$ and $\omega=\omega(\xi_0)$ is found via
$$ \omega = {1-2\xi_0 - \sqrt{(1-2\xi_0)^2 - 4\beta\bar\beta}\over 2\bar\beta}\,. $$
Note that binomial bound corresponds to just lower-bounding
$$ \log {1-\xi_0\over \xi_0} + 2\log{1-\omega\over 1+\omega} \ge 0\,. $$
}%

Next, one can notice that the map $(\xi_0, x, R) \mapsto (a_1, a_2)$ is a bijection onto the region
\begin{equation}\label{eq:m5}
	\{(a_1, a_2): 0\le a_1\le 1, 0\le a_2\le a_1 (1-a_1)\}\,.
\end{equation}
With the inverse map given by
\begin{multline*} \xi_0 = {a_2\over 1-a_1 + a_2}, x = {2a_2^2\over 1-a_1 + a_2}, \\
	R=\log2 - \xi_0 h(a_1)-(1-\xi_0)h(a_2)\,.
\end{multline*}

Thus, verifying~\eqref{eq:m4} can as well be done for all $a_1,a_2$ inside the region~\eqref{eq:m5}. Substituting
$g=g_1$ into~\eqref{eq:m4} we get that monotonicity in~\eqref{eq:m3} is equivalent to a two-dimensional inequality:
\begin{multline}\label{eq:m6}
	-2\log {1-a_2\over a_1} \cdot (a_1^2 - a_2^2) \\
		\ge (2a_1^2 - {4\over3}(a_1^3 - a_2^3) - 1) \log {1-a_2\over
a-2}{a_1\over 1-a_1}\,.
\end{multline}
It is possible to verify numerically that indeed~\eqref{eq:m6} holds on the set~\eqref{eq:m5}. For example, one may
first demonstrate that it is sufficient to restrict to $a_2 = 0$ and then verify a corresponding inequality in $a_1$
only. We omit mechanical details.
\end{proof}

\centerline{\sc Acknowledgement}
We thank Prof. A. Barg for reading and commenting on an earlier draft and anonymous reviewers for pointing out a mistake in
the previous version of Table~\ref{tab:comp} and for simplifying proof of~\eqref{eq:n2a}.

\begin{IEEEbiographynophoto}{Yury Polyanskiy}(S'08-M'10-SM'14) is an 
Associate Professor of Electrical Engineering and Computer Science and a member of LIDS at MIT.
Yury received M.S. degree in applied mathematics and physics from the Moscow Institute of Physics and Technology,
Moscow, Russia in 2005 and Ph.D. degree in electrical engineering from Princeton
University, Princeton, NJ in 2010. In 2000-2005 he lead the development of the embedded software in the
Department of Oilfield Surface Equipment, Borets Company LLC (Moscow). Currently, his research focuses on basic questions in information theory, error-correcting codes, wireless communication and fault-tolerant and defect-tolerant circuits.
Dr. Polyanskiy won the 2013 NSF CAREER award and 2011 IEEE Information Theory Society Paper Award.
\end{IEEEbiographynophoto}


\end{document}